\documentclass[12pt]{article}
\usepackage{amssymb,amsthm,amsmath,latexsym}
\usepackage{url}
\usepackage{color}
\usepackage{comment}
\newtheorem{thm}{Theorem}[section]
\newtheorem{prop}[thm]{Proposition}
\newtheorem{lem}[thm]{Lemma}

\theoremstyle{remark}
\newtheorem{rem}[thm]{Remark}
\newcommand{\FF}{\mathbb{F}}
\newcommand{\ZZ}{\mathbb{Z}}
\newcommand{\0}{\mathbf{0}}
\newcommand{\1}{\mathbf{1}}

\DeclareMathOperator{\wt}{wt}

\begin{document}
\title{On the minimum weights of binary LCD codes and ternary LCD codes
}

\author{
Makoto Araya\thanks{Department of Computer Science,
Shizuoka University,
Hamamatsu 432--8011, Japan.
email: {\tt araya@inf.shizuoka.ac.jp}},
Masaaki Harada\thanks{
Research Center for Pure and Applied Mathematics,
Graduate School of Information Sciences,
Tohoku University, Sendai 980--8579, Japan.
email: {\tt mharada@tohoku.ac.jp}.}
and 
Ken Saito\thanks{
Research Center for Pure and Applied Mathematics,
Graduate School of Information Sciences,
Tohoku University, Sendai 980--8579, Japan.
email: {\tt kensaito@ims.is.tohoku.ac.jp}.}
}

%\date{}

\maketitle

\begin{abstract}
Linear complementary dual (LCD) codes are linear codes that intersect
with their dual codes trivially.
We study the largest minimum weight
$d_2(n,k)$ among all binary LCD $[n,k]$ codes and
the largest minimum weight
$d_3(n,k)$ among all ternary LCD $[n,k]$ codes.
The largest minimum weights $d_2(n,5)$ and $d_3(n,4)$ are
partially determined.
We also determine the largest minimum weights 
$d_2(n,n-5)$, $d_3(n,n-i)$ for $i \in \{2,3,4\}$,
and $d_3(n,k)$ for $n \in \{11,12,\ldots,19\}$.
\end{abstract}

%%%%%%%%%%%%%%%%%%%%%%%%%%%%%%%%%%
\section{Introduction}\label{sec:1}
Linear complementary dual (LCD for short) codes are linear codes that intersect
with their dual codes trivially.
LCD codes were introduced by Massey~\cite{Massey} and gave an optimum linear
coding solution for the two user binary adder channel.
Recently, much work has been done concerning LCD codes
for both theoretical and practical reasons
(see e.g.~\cite{AH-C}, \cite{AH}, \cite{AHS2}, \cite{CG},
\cite{CMTQ}, \cite{CMTQP}, \cite{DKOSS}, \cite{GKLRW}, \cite{HS},
\cite{Sok} and the references given therein).
In particular, we emphasize the recent work by
Carlet, Mesnager, Tang, Qi and Pellikaan~\cite{CMTQP}.
It has been shown in~\cite{CMTQP} that
any code over $\FF_q$ is equivalent to some LCD code
for $q \ge 4$, where $\FF_q$ denotes the finite field of order $q$.
This motivates us to study binary LCD codes and ternary LCD codes.
%For the elements of $\FF_2$ and $\FF_3$, we take the sets $\{0,1\}$
%and $\{0,1,2\}$, respectively.

It is a fundamental problem to determine the largest minimum weights
$d_2(n,k)$ (resp.\ $d_3(n,k)$)
among all binary (resp.\ ternary) LCD
$[n,k]$ codes for a given pair $(n,k)$.
For arbitrary $n$,
the largest minimum weights $d_2(n,2)$ and $d_2(n,3)$ were determined
in~\cite{GKLRW} and~\cite{HS}, respectively.
Very recently, by considering the simplex codes,
a characterization of LCD codes
having large minimum weights
has been given by the authors~\cite{AHS2}.
Using the characterization,
the largest minimum weights 
$d_2(n,4)$, $d_3(n,2)$ and $d_3(n,3)$
have been determined in~\cite{AHS2} for arbitrary $n$
% for $(q,k) \in  \{(2,4),(3,2),(3,3)\}$
(see also~\cite{PZK} for $d_3(n,2)$).
As a contribution in this direction,
%%by giving a slight improvement of the characterization,
this paper studies the largest minimum weights 
$d_2(n,5)$ and $d_3(n,4)$.
We show the nonexistence of certain LCD codes meeting
the Griesmer bound.
This is a powerful tool for our study of $d_2(n,5)$ and $d_3(n,4)$.
The minimum weights $d_2(n,k)$ were determined for
arbitrary $n$ and $k=n-1$ in~\cite{DKOSS} and
for arbitrary $n$ and $k \in \{n-2,n-3,n-4\}$ in~\cite{AH}.
For arbitrary $n$,
the classifications of ternary LCD $[n,1]$ codes and
ternary LCD $[n,n-1]$ codes were done in~\cite{AH-C}.
The classification of 
ternary LCD $[n,k]$ codes  was also done in~\cite{AH-C} for $n \le 10$.
In this paper, we determine $d_2(n,n-5)$ and 
$d_3(n,n-i)$ $(i \in \{2,3,4\})$ for arbitrary $n$.
We also determine $d_3(n,k)$ for $n \in \{11,12,\ldots,19\}$.
Note that $d_2(n,k)$ is known for $n \le 24$ 
(see~\cite[Table~15]{AH}, \cite[Table~1]{GKLRW} and~\cite[Table~3]{HS}).

This paper is organized as follows.
In Section~\ref{sec:2}, 
we give some definitions, notations and basic results used in this paper.
We review the characterization of LCD codes in~\cite{AHS2}
(Proposition~\ref{prop:main2}).
In Section~\ref{sec:Gb}, we show that there is no certain LCD code meeting
the Griesmer bound (Lemma~\ref{lem:Gb}).
By using Lemma~\ref{lem:Gb},
it is shown that
$d_2((2^k-1)s+k+1,k)=
2^{k-1}s+2$  if $k$ is even with $k \ge 4$
and $2^{k-1}s+1$  if $k$ is odd  with $k \ge 3$
for a positive integer $s$ (Proposition~\ref{prop:sk-1}).
In addition, by Lemma~\ref{lem:Gb}, we determine the
largest minimum weights
$d_2(127s+15,7)$,
$d_2(127s+23,7)$ and $d_2(511s+17,9)$
for a nonnegative integer $s$ (Proposition~\ref{prop:F2}).
Proposition~\ref{prop:main2} and Lemma~\ref{lem:Gb}
are powerful tools for our study of $d_2(n,5)$ and $d_3(n,4)$
in Sections~\ref{sec:2-5} and~\ref{sec:3-4}.
In Section~\ref{sec:2-5},
%we study the largest minimum weights $d_2(n,5)$.
we determine the largest minimum weights $d_2(n,5)$ for
\[
n \equiv 0,1,6,9,13,15,17,21,23,24,25,27,28,29,30 \pmod{31}
\]
(Theorem~\ref{thm:2-5}).
For the remaining lengths, we also give bounds on $d_2(n,5)$.
In Section~\ref{sec:3-4},
%we study the largest minimum weights $d_3(n,4)$.
we determine the largest minimum weights $d_3(n,4)$ for
\begin{multline*}
 n \equiv 4,5,6,7,8,10,11,14,16,17,19,20,24,\\
 26,27,29,30,33,35,36,
38,39 \pmod{40}
\end{multline*}
(Theorem~\ref{thm:3-4}).
For the remaining lengths, we also give bounds on $d_3(n,4)$.
In Section~\ref{sec:d}, we determine
the largest minimum weights $d_3(n,k)$ for $n \in \{11,12,\ldots,19\}$.
% and $d_3(20,k)$ with $5$ exceptions.
%% As a consequence, we determine $d_3(n,k)$, when
%% \[
%% (n,k) \in\{(121s+17,5),
%% (364s+13,6),
%% (364s+18,6),
%% (1093s+14,7)\}
%% \]
%% for a nonnegative integer $s$.
As a consequence, we determine the largest minimum weights $d_3(n,k)$,
where
\[
(n,k) \in 
\left\{
\begin{array}{l}
( 121s+ 11, 5),
( 364s+ 12, 6),
(364s+13,6),   \\
(1093s+ 13, 7),
(1093s+14,7),  
(3280s+ 14, 8),\\
(121s+ 15, 5),
(121s+17,5),
(364s+18,6)
\end{array}
\right\}
\]
% d_3(  40s+ 10, 4)=   27s+ 5,
for a nonnegative integer $s$.
% We also give some bounds on $d_3(n,k)$ for $n \in \{19,20\}$.
Finally, in Section~\ref{sec:large},
we examine the largest minimum weights $d_2(n,n-i)$ and 
$d_3(n,n-i)$ for small $i$ and arbitrary $n$.
In particular, we completely determine
$d_2(n,n-5)$ and $d_3(n,n-i)$ for $i\in\{2,3,4\}$ and arbitrary $n$.

All computer calculations in this paper were
done by programs in the language C
and
programs in {\sc Magma}~\cite{Magma}.

%%%%%%%%%%%%%%%%%%%%%%%%%%%%%%%%%%%%%%
\section{Preliminaries}\label{sec:2}
In this section, 
we give some definitions, notations and basic results used in this
paper.
Lemma~\ref{lem:extend} is an important method for constructing
LCD codes with large minimum weights, which is used
throughout this paper.
Proposition~\ref{prop:main2} is a powerful tool for our study
in Sections~\ref{sec:2-5} and~\ref{sec:3-4}.

%%%%%%%%%%%%%%%%%%%%%%%%%%%%%%
\subsection{Definitions and notations}

Let $\FF_q$ denote the finite field of order $q$,
where $q$ is a prime power.
A linear $[n,k]$ code over $\FF_q$ is a $k$-dimensional vector 
subspace of $\FF_q^n$.
All codes in this paper are linear, and 
codes means are linear codes.
Codes over $\FF_2$ and $\FF_3$ are called {\em binary} and
{\em ternary}, respectively.
The {\em weight} $\wt(x)$ of a vector $x \in \FF_q^n$ is
the number of non-zero components of $x$.
A vector of an $[n,k]$ code $C$ over $\FF_q$
is called a {\em codeword}.
The minimum non-zero weight of all codewords in $C$ is called
the {\em minimum weight} of $C$.
An $[n,k]$ code with minimum weight $d$ is called an $[n,k,d]$ code.
Two codes $C$ and $C'$ over $\FF_q$ are {\em equivalent} 
if there is a monomial matrix $P$ with 
$C' = C \cdot P$, where
$C \cdot P = \{ x P \mid x \in C\}$.
For any $[n,k,d]$ code over $\mathbb{F}_q$, it is known that
$
n\ge \sum_{i=0}^{k-1} \left\lceil \frac{d}{q^i}\right\rceil$.
This bound is well known as the {\em Griesmer bound}.
For a given set of parameters $q,n,k$,
define
\[g_q(n,k)=\max\left\{d \in \mathbb{Z}_{\ge 0}~\middle|~ n
\ge \sum_{i=0}^{k-1} \left\lceil \frac{d}{q^i}\right\rceil\right\},\]
where $\ZZ_{\ge 0}$ denotes the set of nonnegative integers.

The {\em dual} code $C^{\perp}$ of an $[n,k]$ code $C$ 
over $\FF_q$ is defined as
$
C^{\perp}=
\{x \in \FF_q^n \mid \langle x,y\rangle = 0 \text{ for all } y \in C\},
$
where $\langle x,y\rangle = \sum_{i=1}^{n} x_i {y_i}$
for $x=(x_1,x_2,\ldots,x_n), y=(y_1,y_2,\ldots,y_n) \in \FF_q^n$.
The minimum weight of the dual code $C^\perp$ of a code $C$ is called
the {\em dual distance} of $C$ and it is denoted by $d^\perp$.
A code $C$ is {\em self-orthogonal} if $C \subset C^\perp$.

A code $C$ over $\FF_q$ is called 
{\em linear complementary dual} (LCD for short)
if $C \cap C^\perp = \{\0_n\}$, 
where  $\0_n$ denotes the zero vector of length $n$.
Let $d_q(n,k)$ denote the largest minimum weight among all LCD $[n,k]$
codes over $\FF_q$.
%% By the Griesmer bound, any $[n,k]$ code over $\mathbb{F}_q$ has minimum
%% weight at most $g_q(n,k)$.
%% Let $N_q(n,k)$ denote the number of all inequivalent LCD
%% $[n,k,d_q(n,k)]$ codes over $\FF_q$.
%% An LCD $[n,k,d]$ code over $\FF_q$ is called {\em optimal}
%% if there is no LCD $[n,k,d']$ code over $\FF_q$ for all $d' > d$.
%% In other words, an optimal LCD $[n,k]$ code over $\FF_q$ has minimum weight
%% $d_q(n,k)$.
%% Let $C$ be an LCD $[n,k,d]$ code over $\FF_q$.
%% Define an $[n+1,k,d]$ code $\overline{C}$ over $\FF_q$ as follows:
%% $\overline{C}=\{(x,0) \mid x \in C\}$.
%% 
%% \begin{lem}[{\cite[Proposition~3]{AH-C}}]\label{lem:class}
%% Let $\cC^q_{n,k,d}$ denote all equivalence classes of LCD $[n,k,d]$
%% codes over $\FF_q$.
%% Let $\cD^q_{n,k,d}$ denote all equivalence classes of LCD $[n,k,d]$
%% codes over $\FF_q$ with dual distances $d^\perp\ge 2$.
%% Let $\overline{\cC^q_{n-1,k,d}}$ denote 
%% all classes containing $\overline{C_1},\overline{C_2},\ldots,
%% \overline{C_t}$,
%% where
%% $C_1,C_2,\ldots,C_t$ denote representatives of
%% $\cC^q_{n-1,k,d}$, and
%% $t=|\overline{\cC^q_{n-1,k,d}}|$.
%% Then 
%% $\cC^q_{n,k,d}= \cD^q_{n,k,d} \cup \overline{\cC^q_{n-1,k,d}}$.
%% \end{lem}
The following characterization is due to Massey~\cite{Massey}.

\begin{prop}
Let $C$ be a code over $\FF_q$ with generator matrix $G$.
Then the following properties are equivalent:
\begin{enumerate}
\item $C$ is LCD,
\item $C^\perp$ is LCD,
\item $GG^T$ is nonsingular,
where
$G^T$  denotes the transpose of a matrix $G$.
\end{enumerate}
\end{prop}

Throughout this paper, we use the above proposition
without mentioning this,
when we determine whether a given code is LCD or not.

%%%%%%%%%%%%%%%%%%%%%%%%%%%%%%
\subsection{Simplex codes and codes $C_{q,k}(m)$}

In this subsection, we give background materials used in
Sections~\ref{sec:2-5} and~\ref{sec:3-4}.
Suppose that $(q,k_0) \in \{(2,3),(3,2)\}$.
We use the following notation:
\[
[k]_q=\frac{q^k-1}{q-1}
\]
for a positive integer $k$.
%For $k \ge 2$,
By induction,
we define the $k \times [k]_q$ $\mathbb{F}_q$-matrices $S_{q,k}$ as follows:
\begin{align*}
&
S_{2,1}=
\begin{pmatrix}
1
\end{pmatrix} \text{ and }
S_{2,k}=
\left(
\begin{array}{ccc}
S_{2,k-1} & \0_{k-1}^T & S_{2,k-1} \\ 
\0_{[k-1]_2} & 1 & \1_{[k-1]_2}
\end{array}
\right)  \text{ if } k \ge 2,\\
&
 S_{3,1}=
\begin{pmatrix}
1
\end{pmatrix} \text{ and }
S_{3,k}=
\left(
\begin{array}{cccc}
S_{3,k-1} & \0_{k-1}^T & S_{3,k-1} & S_{3,k-1} \\ 
\0_{[k-1]_3} & 1 & \1_{[k-1]_3} & 2\1_{[k-1]_3}
\end{array}
\right)   \text{ if } k \ge 2,
\end{align*}
where $\1_n$ denotes the all-one vector of length $n$.
The matrix $S_{q,k}$ is a generator matrix of the simplex $[[k]_q,k,q^{k-1}]$ code.
The simplex $[[k]_q,k,q^{k-1}]$ code is a constant weight code~\cite[Theorem~2.7.5]{HP}.
The simplex $[[k]_q,k,q^{k-1}]$ code is self-orthogonal if $k \ge k_0$~\cite[Theorems~1.4.8~(ii) and 1.4.10~(i)]{HP}.
%Note that the simplex code over $\mathbb{F}_q$ is not necessarily self-orthogonal if $q \ge 4$. 

Throughout this paper, 
$A^{(s)}$ denotes the juxtaposition $(A \cdots A)$
of $s$-copies of $A$ for a matrix $A$.
The following method for construction LCD codes is used 
throughout this paper.

\begin{lem}[{\cite[Lemma~3.5]{AHS2}}]
\label{lem:extend}
Suppose that $(q,k_0) \in \{(2,3),(3,2)\}$, 
$k \ge k_0$, $s$ is a positive integer.
Let $C$ be an LCD $[n,k,d]$ code over $\FF_q$
with generator matrix $G$.
Then the code with generator matrix of form:
\[
\left(\begin{array}{cc}
 S_{q,k}^{(s)} &G
       \end{array}\right)
\]
%% is an LCD $[n,k,d]$ code, where
%% \[
%% n=\sum_{i=1}^{[k]_q}m_i+[k]_q \cdot s \text{ and  }d=d_0+q^{k-1}s.
%% \]
is an LCD $[n+[k]_q \cdot s,k,d+q^{k-1}s]$ code over $\FF_q$.
\end{lem}

Let $h_{q,k,i}$ be the $i$-th column of
the $k \times [k]_q$ $\mathbb{F}_q$-matrices $S_{q,k}$.
For a vector $m=(m_1,m_2,\ldots,m_{[k]_q}) \in \mathbb{Z}_{\ge 0}^{[k]_q}$, we define the $k \times \sum_{i=1}^{[k]_q}m_i$ $\mathbb{F}_q$-matrix $G_{q,k}(m)$, which consists of $m_i$ columns $h_{q,k,i}$ for each $i$ as follows:
\begin{equation}\label{eq:Gqkm}
G_{q,k}(m)=
\left(
h_{q,k,1} \cdots h_{q,k,1} \cdots h_{q,k,[k]_q} \cdots h_{q,k,[k]_q}
\right).
\end{equation}
Here, we remark that $m_i=0$ means no column of $G_{q,k}(m)$ is $h_{q,k,i}$.
We denote by $C_{q,k}(m)$ the code with generator matrix $G_{q,k}(m)$.

%%%%%%%%%%%%%%%%%%%%%%%%%%%%%%
% \subsection{Some results on LCD codes in~\cite{AHS2}}
  
% Now we review the results on LCD codes given in~\cite{AHS2}.

% \begin{lem}[{\cite[Lemma~3.3]{AHS2}}]
% \label{lem:extend}
% Suppose that $(q,k_0) \in \{(2,3),(3,2)\}$, 
% $k \ge k_0$, $s$ is a positive integer and 
% $m=(m_1,m_2,\ldots,m_{[k]_q}) \in \mathbb{Z}_{\ge 0}^{[k]_q}$.
% %% If $C_{q,k}(m)$ is an LCD
% %% $[\sum_{i=1}^{[k]_q}m_i,k,d_0]$
% %% code, then the code with generator matrix
% Let $C_{q,k}(m)$ be the $[\sum_{i=1}^{[k]_q}m_i,k,d]$ code over $\FF_q$
% with generator matrix $G_{q,k}(m)$ of form~\eqref{eq:Gqkm}.
% If $C_{q,k}(m)$ is an LCD code,
% then the code with generator matrix of form:
% \[
% \left(\begin{array}{cc}
%  S_{q,k}^{(s)} &G_{q,k}(m)
%        \end{array}\right)
% \]
% %% is an LCD $[n,k,d]$ code, where
% %% \[
% %% n=\sum_{i=1}^{[k]_q}m_i+[k]_q \cdot s \text{ and  }d=d_0+q^{k-1}s.
% %% \]
%  is an LCD $[\sum_{i=1}^{[k]_q}m_i+[k]_q \cdot s,k,d+q^{k-1}s]$ code
%  over $\FF_q$.
% \end{lem}
% 
% The above method for construction LCD codes is used 
% throughout this paper.

%The following bound~\eqref{eq:mi} is useful for showing that there is no
%LCD $[n,k,d]$ code over $\FF_q$ with dual distance $d^\perp \ge 2$ for
%a given $(q,n,k,d)$.

\begin{lem}[{\cite[Lemma~3.6]{AHS2}}]
\label{lem:mi}
Suppose that $(q,k_0) \in \{(2,3),(3,2)\}$, 
$k \ge k_0$ and 
$m=(m_1,m_2,\ldots,m_{[k]_q}) \in \mathbb{Z}_{\ge 0}^{[k]_q}$.
Let $C_{q,k}(m)$ be the $[n,k]$ code over $\FF_q$
with generator matrix $G_{q,k}(m)$ of form~\eqref{eq:Gqkm},
where $n=\sum_{i=1}^{[k]_q}m_i$.
If $C_{q,k}(m)$ is an LCD code with minimum weight at least $d$,
then
\begin{equation}\label{eq:mi}
 qd-(q-1)n \le m_i \le n-\frac{q^{k-1}-1}{(q-1)q^{k-2}}d
 \end{equation}
for each $i \in \{1,2,\ldots,[k]_q\}$.
 \end{lem}

\begin{rem}\label{rem}
By considering all vectors $m=(m_1,m_2,\ldots,m_{[k]_q}) \in
\mathbb{Z}_{\ge 0}^{[k]_q}$ 
such that $n=\sum_{i=1}^{[k]_q}m_i$ and~\eqref{eq:mi},
it is possible to find
representatives of all equivalence classes of LCD
 $[n,k]$ codes over $\FF_q$ with
minimum weights at least $d$ and  dual distances $d^\perp \ge 2$
as $C_{q,k}(m)$ for a given set of parameters $q,n,k,d$.
\end{rem}

% \begin{prop}
% \label{prop:190127-1}
% Suppose that $(q,k_0) \in \{(2,3),(3,2)\}$,
% $k \ge k_0$ and $n \equiv 0 \pmod{[k]_q}$.
% Then there is no LCD $[n,k,g_q(n,k)]$ code over $\mathbb{F}_q$.
%  \end{prop}

%% We assume that $(q,k_0) \in \{(2,3),(3,2)\}$ and $qd-(q-1)n \ge 1$.

% Suppose that $k \ge k_0$ and $qd-(q-1)n \ge 1$.

%% \begin{thm}
%% \label{thm:main}
%% Suppose that $(q,k_0) \in \{(2,3),(3,2)\}$, $qd-(q-1)n \ge 1$
%% and $qr_{q,n,k,d} \ge k$, where
%% $r_{q,n,k,d}$ is the integer defined in \eqref{eq:r}.
%% For $k \ge k_0$, 
%% there is a one-to-one correspondence between
%% equivalence classes of LCD $[n,k,d]$ codes over $\FF_q$
%% with dual distances $d^\perp \ge 2$ 
%% and 
%% equivalence classes of LCD
%% $[qr_{q,n,k,d},k,(q-1)r_{q,n,k,d}]$ codes over $\FF_q$
%% with dual distances $d^\perp \ge 2$.
%% %where $r_{q,n,k,d}=q^{k-1}n-[k]_q \cdot d$.
%% \end{thm}
%Fix $q,n,k$ and $d$.

In Sections~\ref{sec:2-5} and~\ref{sec:3-4},
we study the largest minimum weights $d_2(n,5)$ and $d_3(n,4)$
by using the following proposition.

\begin{prop}[{\cite[Theorem~4.7]{AHS2}}]
\label{prop:main2}
Suppose that $(q,k_0) \in \{(2,3),(3,2)\}$ and $k \ge k_0$.
Assume that we write
\[
n=[k]_q \cdot s+t,
\]
where $s \in \ZZ_{\ge 0}$
and $t \in \{0,1,\ldots,[k]_q-1\}$.
In addition, assume the following:
\begin{equation}\label{eq:as}
\begin{split}
&\text{the minimum weight $d$ is written as}\\
&d(s,t)=q^{k-1}s+\alpha(t), \\
&\text{where $\alpha(t)$ is a constant depending on only $t$.}
\end{split}
\end{equation}
Let $r$ and $s'$ denote the integers
$r_{q,([k]_q \cdot s+t) ,k,d(s,t)}$ and $s'_{q,([k]_q \cdot s+t),k,d(s,t)}$,
where
\begin{align}\label{eq:r}
r_{q,([k]_q \cdot s+t),k,d(s,t)}&=q^{k-1}([k]_q \cdot s+t)-[k]_q \cdot d(s,t),\\
\label{eq:s0}
s'_{q,([k]_q \cdot s+t),k,d(s,t)}&= \frac{qr_{q,([k]_q \cdot s+t),k,d(s,t)}-t}{[k]_q}+1.
\end{align}
If $q r \ge k$ and 
there is no LCD code over $\FF_q$
with dual distances $d^\perp \ge 2$ and parameters
\[
[qr,k,(q-1)r]
=
[[k]_q \cdot (s'-1)+t,k,q^{k-1}(s'-1)+\alpha(t)],
\]
then
     there is no  LCD code over $\FF_q$ with parameters
\[
[[k]_q \cdot s+t,k,q^{k-1}s+\alpha(t)]
\]
for every integer $s$.
\end{prop}

We remark that the assumption~\eqref{eq:as} is automatically satisfied
for our study in Sections~\ref{sec:2-5} and~\ref{sec:3-4}.

% \begin{center}
%  {\bf +++++ under investigation +++++}
% \end{center} 
%Suppose that $(q,k_0) \in \{(2,3),(3,2)\}$ and $k \ge k_0$.
%Proposition~\ref{prop:main} claims that
%if there is no LCD
%$[[k]_q (s'_{q,n,k,d}-1),k,(q-1)r_{q,n,k,d}]$ code
%over $\mathbb{F}_q$ with $d^\perp \ge 2$,
%then there is no LCD $[[k]_q \cdot s+t,k,d]$ code
%over $\mathbb{F}_q$
%for $s \ge s'_{q,n,k,d}$.
% 
%\begin{prop}\label{prop:main2}
%If there is no LCD $[qr_{q,n,k,d},k,(q-1)r_{q,n,k,d}]$ code
%over $\mathbb{F}_q$ with $d^\perp \ge 2$,
%then there is no LCD $[[k]_q \cdot s+t,k,d]$ code over $\mathbb{F}_q$
%for an integer $s$ with $[k]_q \cdot s+t \ge k$.
%\end{prop}
%\begin{proof}
%By Proposition~\ref{prop:main},
%it is sufficient to show the case $s_{q,n,k,d} < s'_{q,n,k,d}$.
%\end{proof}
% \begin{center}
%  {\bf +++++ under investigation +++++}
% \end{center} 

%%%%%%%%%%%%%%%%%%%%%%%%%%%%
\section{LCD codes meeting the Griesmer bound}
\label{sec:Gb}

Ward~\cite{W98} studied the divisibility of codes meeting the Griesmer
bound.
The following lemma is a consequence of~\cite[Theorem~1]{W98}, however,
it is a powerful tool for our study in
Sections~\ref{sec:2-5} and~\ref{sec:3-4}.
%% More precisely, the lemma is used to establish the nonexistence of
%% certain binary LCD codes of dimension $5$ and
%% ternary LCD codes of dimension $4$
%% meeting the Griesmer bound.

\begin{lem}\label{lem:Gb}
\begin{enumerate}
\item 
If $k$ is odd and $d$ is even, then there is no binary LCD $[n,k,d]$ code
meeting the Griesmer bound for every positive integer $n$.
\item
If $d \equiv 0 \pmod 3$, then there is no ternary LCD $[n,k,d]$ code
meeting the Griesmer bound for every positive integers $n$ and $k$.
\end{enumerate}
\end{lem}
\begin{proof}
\begin{enumerate}
\item 
Suppose that there is a binary LCD $[n,k,d]$ code $C$
meeting the Griesmer bound.
If $d$ is even, then  
$C$ is an even code~\cite[Theorem~1]{W98}.
Any binary even LCD code must have even dimension~\cite[Theorem~5]{CMTQ}.

\item 
Suppose that there is a ternary LCD $[n,k,d]$ code $C$
meeting the Griesmer bound and $d \equiv 0 \pmod 3$.
By~\cite[Theorem~1]{W98},
the weight of any codeword in $C$ is a multiple of $3$.
It is known that the weights of all codewords in $C$ are 
multiples of $3$
if and only if $C$ is self-orthogonal
(see e.g.~the proof of~\cite[Theorem~2.1]{KP}).
This is a contradiction.
\end{enumerate}
This completes the proof. 
\end{proof}

\begin{rem}
A similar argument shows the nonexistence of a
binary LCD $[n,k,d]$ code
meeting the Griesmer bound for $d \equiv 0 \pmod 4$
(see the proof of~\cite[Proposition~3.9]{DKOSS}).
\end{rem}  

%%%%%%%%%%%%%%%%%%%%%%%%%%%%%%%%%%%%%%%%%

It was shown that $d_2(k+1,k)=2$ and $1$ if
$k$ is even and odd, respectively~\cite[Proposition~3.2]{DKOSS}.
This result is generalized, as an example of the above lemma.

\begin{prop}\label{prop:sk-1}
For a positive integer $k \ge 3$ and a positive integer $s$,
\[
d_2((2^k-1)s+k+1,k)=
\begin{cases}
2^{k-1}s+2 \text{ if $k$ is even,}\\ 
2^{k-1}s+1 \text{ if $k$ is odd.} 
\end{cases}
\]
\end{prop}
%% \begin{proof}
%% Suppose that $k$ is even.
%% There is a binary LCD $[k+1,k,2]$ code~\cite[Proposition 3.2]{DKOSS}.
%% By Lemma~\ref{lem:extend}, a binary LCD
%% $[(2^k-1)s+k+1,k,2^{k-1}s+2]$ code meeting the Griesmer bound is constructed
%% for every positive integer $s$.
%% 
%% Suppose that $k$ is odd.
%% By Lemma~\ref{lem:Gb},
%% there is no binary LCD
%% $[(2^k-1)s+k+1,k,2^{k-1}s+2]$ code meeting the Griesmer bound
%% for every positive integer $s$.
%% There is a binary LCD $[k+1,k,1]$ code~\cite[Proposition 3.2]{DKOSS}.
%% By Lemma~\ref{lem:extend},
%% a binary LCD $[(2^k-1)s+k+1,k,2^{k-1}s+1]$ code is constructed.
%% \end{proof}
\begin{proof}
Suppose that $k \ge 3$.
Then, by Lemma~\ref{lem:extend}, there is a binary LCD
$[(2^k-1)s+k+1,k,2^{k-1}s+d_2(k+1,k)]$ code for every positive integer $s$.
For every positive integer $s$,
the Griesmer bound is the same as the following bound:
\[
d_2((2^k-1)s+k+1,k) \le  2^{k-1}s+2.
\]
Note that the equality holds in the above bound if and only if
the equality holds in the Griesmer bound.
When $k$ is odd,
by Lemma~\ref{lem:Gb}~(i),
there is no binary LCD
$[(2^k-1)s+k+1,k,2^{k-1}s+2]$ code for every positive integer $s$.
The result follows.
\end{proof}

In addition, by Lemma~\ref{lem:Gb}~(i), we determine the
largest minimum weights
$d_2(127s+15,7)$,
$d_2(127s+23,7)$ and $d_2(511s+17,9)$
for a nonnegative integer $s$.

\begin{prop}\label{prop:F2}
For a nonnegative integer $s$,
\[
\begin{array}{ll}
d_2(127s+15,7)=64s+5, &
d_2(127s+23,7)=64s+9,  \\
d_2(511s+17,9)=256s+5.
\end{array}
\]
\end{prop}
\begin{proof}
Let $s$ be a positive integer.
By the Griesmer bound, we have
\[
\begin{array}{ll}
d_2(127s+15,7)\le 64s+6, &
d_2(127s+23,7)\le 64s+10, \\ 
d_2(511s+17,9)\le 256s+6.
\end{array}
\]
For 
\begin{equation*}\label{eq:F2}
(n,k,d) \in
\left\{
\begin{array}{l}
(127s+15,7,64s+6),   
(127s+23,7,64s+10), \\ 
(511s+17,9,256s+6)
\end{array}
\right\},
\end{equation*}
 each binary $[n,k,d]$ code meets the Griesmer bound.
Since $k$ is odd and $d$ is even,
by Lemma~\ref{lem:Gb}~(i), it is not LCD.
Hence, we have
\[
\begin{array}{ll}
d_2(127s+15,7)\le 64s+5, &
d_2(127s+23,7)\le 64s+9,   \\
d_2(511s+17,9)\le 256s+5.
\end{array}
\]
It is known that 
$d_2( 15, 7)= 5$, 
$d_2( 17, 9)= 5$ and $d_2(23, 7)= 9$~\cite[Table~15]{AH}
 and~\cite[Table~3]{HS}.
By Lemma~\ref{lem:extend},
there is a binary LCD
$[n,k,d]$ code for
\[
(n,k,d) \in
\left\{
\begin{array}{l}
(127s+15,7,64s+5),   
(127s+23,7,64s+9), \\ 
(511s+17,9,256s+5)
\end{array}
\right\}.
\]
This completes the proof. 
\end{proof}

\begin{rem}
Note that $d_2(n,k)$ is known for $n \le 24$
(see~\cite[Table~15]{AH}, \cite[Table~1]{GKLRW} and~\cite[Table~3]{HS}).
Only the parameters
$[15, 7, 6]$, $[17, 9, 6]$ and $[23, 7,10]$ are
parameters $[n,k,d_2(n,k)+1]$ meeting the Griesmer bound and
satisfying the
assumption of Lemma~\ref{lem:Gb}~(i) for $n \le 24$ and $6 \le k \le n-6$.
\end{rem}

%%%%%%%%%%%%%%%%%%%%%%%%%%%%
\section{Binary LCD codes of dimension 5}
\label{sec:2-5}

In this section, we study the largest minimum weights $d_2(n,5)$.
For $n \ge 5$,
write $n=31s+t$, where $s \in \ZZ_{\ge 0}$ and $t \in \{0,1,\ldots,30\}$.
We list $g_2(31s+t,5)$ in Table~\ref{Tab:G-bound-2-5}.

%%%%%%%%%%%%%%%%%%%%%%%%%%%%%%%%%%%
\begin{table}[thb]
\caption{$g_2(31s+t,5)$}
\label{Tab:G-bound-2-5}
\begin{center}
{\small
\begin{tabular}{c|c||c|c||c|c}
\noalign{\hrule height0.8pt}
$n$ & $g_2(n,5)$ & $n$ & $g_2(n,5)$ & $n$ & $g_2(n,5)$  \\\hline
$31s$   & $16s$    &$31s+11$ & $16s+4$ &$31s+22$ & $16s+10$ \\
$31s+1$ & $16s$    &$31s+12$ & $16s+5$ &$31s+23$ & $16s+11$ \\
$31s+2$ & $16s$    &$31s+13$ & $16s+6$ &$31s+24$ & $16s+12$ \\
$31s+3$ & $16s$    &$31s+14$ & $16s+6$ &$31s+25$ & $16s+12$ \\
$31s+4$ & $16s$    &$31s+15$ & $16s+7$ &$31s+26$ & $16s+12$ \\
$31s+5$ & $16s+1$  &$31s+16$ & $16s+8$ &$31s+27$ & $16s+13$ \\
$31s+6$ & $16s+2$  &$31s+17$ & $16s+8$ &$31s+28$ & $16s+14$ \\
$31s+7$ & $16s+2$  &$31s+18$ & $16s+8$ &$31s+29$ & $16s+14$ \\
$31s+8$  & $16s+3$ &$31s+19$ & $16s+8$ &$31s+30$ & $16s+15$ \\
$31s+9$  & $16s+4$ &$31s+20$ & $16s+9$ & & \\
$31s+10$ & $16s+4$ &$31s+21$ & $16s+10$& & \\
\noalign{\hrule height0.8pt}
\end{tabular}
}
\end{center}
\end{table}
%%%%%%%%%%%%%%%%%%%%%%%%%%%%%%%%%%%

%%%%%%%%%%%%%%%%
\subsection{Known results on $d_2(n,5)$ and corrections of~\cite{AH}}

It was shown in~\cite[(5)]{AH} that
if $n \equiv 3,5,7,11,19,20,22,26 \pmod{31}$ and $n \ge 5$, then
\[
d_2(n,5)=\left\lfloor \frac{16n}{31}\right\rfloor-1, 
\]
and if $n \equiv 4 \pmod{31}$ and $n \ge 5$, then
\[
d_2(n,5)=\left\lfloor \frac{16n}{31}\right\rfloor-2.
\]
In the course of preparing this paper, we discovered 
some errors of~\cite{AH}.
In~\cite[(5)]{AH}, for $n \equiv 12 \pmod{31}$ the upper bound
was incorrectly stated to be
$d \le \left\lfloor \frac{16n}{31}\right\rfloor-2$
and the correct bound is
$d \le \left\lfloor \frac{16n}{31}\right\rfloor-1$.
This led to the error in~\cite[Proposition~3]{AH} for
$n \equiv 12 \pmod{31}$.
The correct version of~\cite[Proposition~3]{AH} is as follows:

\begin{prop}\label{prop:correct}
If $n \equiv 1,9,13,15,17,21,23,24,25,27,28,29,30 \pmod{31}$
and $n \ge 5$, then
\[
d_2(n,5)=\left\lfloor \frac{16n}{31}\right\rfloor \text{ or }
\left\lfloor \frac{16n}{31}\right\rfloor-1.
\]
If $n \equiv 2,6,8,10,12,14,18 \pmod{31}$
and $n \ge 5$, then
\[
d_2(n,5)=\left\lfloor \frac{16n}{31}\right\rfloor-1 \text{ or }
\left\lfloor \frac{16n}{31}\right\rfloor-2.
\]
 If $n \equiv 0,16 \pmod{31}$ and $n \ge 5$, then
\[
d_2(n,5)=
\left\lfloor \frac{16n}{31}\right\rfloor,
\left\lfloor \frac{16n}{31}\right\rfloor -1
\text{ or }
\left\lfloor \frac{16n}{31}\right\rfloor-2.
\]
\end{prop}

%%%%%%%%%%%%%%%%
\subsection{New results on $d_2(n,5)$}

As described above,
Lemma~\ref{lem:Gb}~(i) is a powerful tool for our study on $d_2(n,5)$.
Lemma~\ref{lem:Gb}~(i) gives the following:
\begin{equation}\label{eq:2-5}
d_2(n,5) \le g_2(n,5)-1 \text{ if } n \equiv 0,6,9,13,16,21,24,28 \pmod{31}
\end{equation}
for $n \ge 5$.

\begin{prop}\label{prop:2-5-1}
For $n \ge 5$,
\begin{align*}
d_2(n,5) = g_2(n,5)-1 \text{ if }  n \equiv 6,9,13,21,24,28 \pmod{31}.
 \end{align*}
\end{prop}
\begin{proof}
It is known that
there is a binary LCD $[n,5,g_2(n,5)-1]$ code
for $n \ge 5$ and  $n \equiv 6,9,13,21,24,28 \pmod{31}$
(see Proposition~\ref{prop:correct}).
The result follows from~\eqref{eq:2-5}.
\end{proof}

In order to apply Proposition~\ref{prop:main2} to this case,
write $n=31s+t$, where $s \in \ZZ_{\ge 0}$ and
\[
t \in \{0,1,2,8,10,12,14,15,16,17,18,23,25,27,29,30\}.
\]
Suppose that
\[
d(s,t)=
\begin{cases}
 g_2(31s+t,5)-1& \text{ if }t \in\{0,16\},\\
 g_2(31s+t,5)  & \text{ otherwise.} 
\end{cases}
\]
Let $r=r_{2,31s+t,5,d(s,t)}$ be the integer defined in~\eqref{eq:r},
where $r$ is listed in Table~\ref{Tab:2-5-0a}.
Note that $d(s,t)$ is written as
$16s+\alpha(t)$, where $\alpha(t)$ is a constant depending on only $t$.
Since $d(s,t)$ satisfies 
the assumption~\eqref{eq:as} in Proposition~\ref{prop:main2},
we have the following:

\begin{prop}\label{prop:2-5-main}
% Let $s'=s'_{2,n,5,d}$ and $r=r_{2,n,5,d}$ be the integers defined in 
% \eqref{eq:s0} and \eqref{eq:r}, respectively.
% For $s \ge s'$,
% if there is no binary LCD $[2r,5,r]$ code
% $C_0$ with $d(C_0^\perp) \ge 2$,
% then there is no binary LCD $[31s+t,5,d]$ code.
If there is no binary LCD $[2r,5,r]$ code
with dual distance $d^\perp \ge 2$,
then there is no binary LCD $[31s+t,5,d(s,t)]$ code for
every integer $s$.
\end{prop}

% \begin{rem}
% For each $n=31s+t$, we list $r$ in Table~\ref{Tab:2-5-0a}.
% \end{rem}

%%%%%%%%%%%%%%%%%%%%%%%%%%%%%%
\begin{table}[thb]
\caption{$r$ in Proposition~\ref{prop:2-5-main}}
\label{Tab:2-5-0a}
\begin{center}
{\small
%{\footnotesize
%{\scriptsize
%{\tiny
\begin{tabular}{c|c||c|c||c|c||c|c}
\noalign{\hrule height0.8pt}
  $n$ & $r$ &  $n$ & $r$ &  $n$ & $r$ &  $n$ & $r$ \\
\hline
$31s$   & $31$ &$31s+10$& $36$ &$31s+16$& $39$ &$31s+25$& $28$ \\
$31s+1$ & $16$ &$31s+12$& $37$ &$31s+17$& $24$ &$31s+27$& $29$ \\
$31s+2$ & $32$ &$31s+14$& $38$ &$31s+18$& $40$ &$31s+29$& $30$ \\
$31s+8$ & $35$ &$31s+15$& $23$ &$31s+23$& $27$ &$31s+30$& $15$ \\
\noalign{\hrule height0.8pt}
\end{tabular}
}
\end{center}
\end{table}
%%%%%%%%%%%%%%%%%%%%%%%%%%%%%%%

By Proposition~\ref{prop:2-5-main}, we examine the nonexistence
of a binary LCD $[2r,5,r]$ code for $r$ in Table~\ref{Tab:2-5-0a}.
It is known that
there is no binary LCD $[30,5,15]$ code~\cite{AH}.
As described in Remark~\ref{rem},
it is possible to find
representatives of all equivalence classes of binary LCD
$[n,5,d]$ codes with dual distances $d^\perp \ge 2$
as $C_{2,5}(m)$, by considering all vectors $m=(m_1,m_2,\ldots,m_{31}) \in
\mathbb{Z}_{\ge 0}^{31}$ 
such that $n=\sum_{i=1}^{31}m_i$ and~\eqref{eq:mi}
for a given set of parameters $n,d$.
Moreover, 
we may assume without loss of generality that 
\[
m_i \ge 1\ (i \in \{1,2,4,8,16\})
\text{ and } \sum_{i \in {\mathcal S}} m_i = d,
\]
where 
% $\Gamma=\{j+k \mid j \in \{1,9,17,25\}, k \in \{0,2,4,6\}\}$.
$\mathcal S$ is the support of the first row of the matrix $S_{2,5}$.
In this way,
our exhaustive computer search  shows that
there is no binary LCD $[2r,5,r]$ code with dual distance $d^\perp \ge 2$
for only $r \in \{16,23,24,27,28,29,30,31\}$, due to the computational complexity.
For reference, 
the time required for the computer search of $r=24$, 
which corresponds to a single core of
a computer with Intel i7,
is approximately $1999$ days.
%% This was done by considering all possible vectors $m$ in
%% generator matrices $G_{2,5}(m)$ of form~\eqref{eq:Gqkm}
%% under the condition~\eqref{eq:mi}.
%\[
%r \in \{16^c,23^c,\mathbf{24,27,28,29,30,31,32,35,36,37,38,39,40}\}.
%\]
For the remaining cases $r$, 
an exhaustive search remains a computational challenge.

By Proposition~\ref{prop:2-5-main},
we have
\[
d_2(31s+t,5)\le
\begin{cases}
g_2(31s+t,5)-1 \text{ if } t \in \{1,15,17,23,25,27,29,30\},\\ 
g_2(31s+t,5)-2 \text{ if } t =0. 
\end{cases}
\]
By Proposition~\ref{prop:correct},
$d_2(31s+t,5) \ge g_2(31s+t,5)-1$  if  $t \in \{1,15,17,23,25,27,29,30\}$
and $d_2(31s+t,5) \ge g_2(31s+t,5)-2$  if  $t=0$.
Therefore, we have the following improvement of Proposition~\ref{prop:correct}.

\begin{thm}\label{thm:2-5} 
If $n \equiv 1,9,13,15,17,21,23,24,25,27,28,29,30 \pmod{31}$
and $n \ge 5$, then
\[
d_2(n,5)=\left\lfloor \frac{16n}{31}\right\rfloor-1.
\]
If $n \equiv 0, 6 \pmod{31}$
and $n \ge 5$, then
\[
d_2(n,5)=\left\lfloor \frac{16n}{31}\right\rfloor-2.
\]
\end{thm}

%%%%%%%%%% old version %%%%%%%%%%%%%%%
% \begin{thm}
% If $n \equiv 1,9,13,15,17,21,23,24,25,27,28,29,30 \pmod{31}$
% and $n \ge 5$, then
% \[
% d_2(n,5)=\left\lfloor \frac{16n}{31}\right\rfloor-1.
% \]
% If $n \equiv 0,2,6,8,10,12,14,16,18 \pmod{31}$
% and $n \ge 5$, then
% \[
% d_2(n,5)=\left\lfloor \frac{16n}{31}\right\rfloor-2.
% \]
% \end{thm}
%%%%%%%%%% old version %%%%%%%%%%%%%%%

% \begin{proof}
% Write $n=31s+t$, where $s \in \ZZ_{\ge 0}$ and
% \[
% t \in \{0,1,2,8,10,12,14,15,16,17,18,23,25,27,29,30\}.
% \]
% Then $d_2(31s+t,5) =g_2(31s+t,5)-2$ if $t \in\{0,16\}$ and
% $d_2(31s+t,5) =g_2(31s+t,5)-1$ otherwise.
% The result follows.
% \end{proof}

% This completely determines the largest minimum weights $d_2(n,5)$.

%%%%%%%%%%%%%%%%%%%%%%%%%%%%%%%%%%%%%%%%%%%%%%%%%%%%%%
\section{Ternary LCD codes of dimension 4}
\label{sec:3-4}

In this section, we study the largest minimum weights $d_3(n,4)$.
For $n \ge 4$,
write $n=40s+t$, where $s \in \ZZ_{\ge 0}$ and $t \in \{0,1,\ldots,39\}$.
We list $g_3(40s+t,4)$ in Table~\ref{Tab:G-bound-3-4}.

%%%%%%%%%%%%%%%%%%%%%%%%%%%%%%%%%%%
\begin{table}[thbp]
\caption{$g_3(40s+t,4)$}
\label{Tab:G-bound-3-4}
\begin{center}
{\small
\begin{tabular}{c|c||c|c||c|c}
\noalign{\hrule height0.8pt}
$n$ & $g_3(n,4)$ & $n$ & $g_3(n,4)$ & $n$ & $g_3(n,4)$  \\\hline
$40s  $ & $27s$    &$40s+14$ & $27s+9$  &$40s+28$ & $27s+18$ \\
$40s+1$ & $27s$    &$40s+15$ & $27s+9$  &$40s+29$ & $27s+18$ \\
$40s+2$ & $27s$    &$40s+16$ & $27s+9$  &$40s+30$ & $27s+19$ \\
$40s+3$ & $27s$    &$40s+17$ & $27s+10$ &$40s+31$ & $27s+20$ \\
$40s+4$ & $27s+1$  &$40s+18$ & $27s+11$ &$40s+32$ & $27s+21$ \\
$40s+5$ & $27s+2$  &$40s+19$ & $27s+12$ &$40s+33$ & $27s+21$ \\
$40s+6$ & $27s+3$  &$40s+20$ & $27s+12$ &$40s+34$ & $27s+22$ \\
$40s+7$ & $27s+3$  &$40s+21$ & $27s+13$ &$40s+35$ & $27s+23$ \\
$40s+8$ & $27s+4$  &$40s+22$ & $27s+14$ &$40s+36$ & $27s+24$ \\
$40s+9$ & $27s+5$  &$40s+23$ & $27s+15$ &$40s+37$ & $27s+24$ \\
$40s+10$ & $27s+6$ &$40s+24$ & $27s+15$ &$40s+38$ & $27s+25$ \\
$40s+11$ & $27s+6$ &$40s+25$ & $27s+16$ &$40s+39$ & $27s+26$ \\
$40s+12$ & $27s+7$ &$40s+26$ & $27s+17$ &&\\
$40s+13$ & $27s+8$ &$40s+27$ & $27s+18$ &&\\
\noalign{\hrule height0.8pt}
\end{tabular}
}
\end{center}
\end{table}
%%%%%%%%%%%%%%%%%%%%%%%%%%%%%%%%%%%

%% $\displaystyle
%% g_3(n,4)=
%% \begin{cases}
%%  \left\lfloor \frac{27n}{40}\right\rfloor \text{ if  } 
%% &n \equiv 
%%  0, 1,10,13,14,19,22,23,25,26, \\
%% &\qquad  27,28,31,32,34,35,36,37,38,39, \\
%%  \left\lfloor \frac{27n}{40}\right\rfloor -1 \text{ if  }
%% &n \equiv
%%  2, 3, 4, 5, 6, 7, 8, 9,11,12, \\
%%  & \qquad 15,16,17,18,20,21,24,29,30,33.
%% \end{cases}
%% $ 

It is known that
there is a ternary LCD $[n,4,d]$ code for
\[
(n,d)\in\{(4,1),(5,2),(6,2),(7,3),(8,4),(10,5)\}
\]
(see~\cite[Table~4 and Proposition~5]{AH-C}).
%% It is trivial that $\FF_3^4$ is the ternary LCD $[4,4,1]$ code.
%% By~\cite[Proposition~5]{AH-C}, there is a ternary LCD $[5,4,2]$ code.
%% It is known~\cite[Table 4]{AH-C} that 
%% there is a ternary LCD $[n,4,d]$ code for
%% $(n,d)\in\{(7,3),(8,4)\}$.
By considering vectors $m$
in generator matrices $G_{3,4}(m)$ of form~\eqref{eq:Gqkm}
such that $n=\sum_{i=1}^{40}m_i$ and~\eqref{eq:mi},
we found a ternary LCD $[n,4,d]$ code $T_n$ for
\[
(n,d) \in 
\left\{
\begin{array}{l}
 (11,6),(14,8),(16,9),(17,10),(19,11),(20,12),(24,15),\\
 (26,16),(29,18), (30,19),(33,21),(35,22),(36,23),\\
 (38,24),(39,25)
\end{array}
\right\}
\]
as codes $C_{3,4}(m)$.
The vectors $m$ are listed in Table~\ref{Tab:3-4-m}.
By Lemma~\ref{lem:extend}, we have the following:

\begin{prop}\label{prop:4-0}
\begin{enumerate}
\item 
If $n \equiv 4,5,7,8,11,16,17,20,24,29,30,33 \pmod{40}$
and $n \ge 4$,
then
there is a ternary LCD $[n,4,g_3(n,4)]$ code.
\item 
If $n \equiv  6,10,14,19,26,35,36,38,39 \pmod{40}$ and $n \ge 4$, then
there is a ternary LCD $[n,4,g_3(n,4)-1]$ code.
\end{enumerate}
 \end{prop}

%%%%%%%%%%%%%%%%%%%%%%%%%%%%%%
\begin{table}[thb]
 \caption{Codes $T_n$ $(n \in \{11,14,16,17,19,20,24,26,29,30,33,
 35,36,38,39\})$}
\label{Tab:3-4-m}
\begin{center}
{\small
%{\footnotesize
%{\scriptsize
%{\tiny
\begin{tabular}{c|c}
\noalign{\hrule height0.8pt}
Codes & Vectors $m$ \\
\hline
$T_{11}$&$(1100110000010101001000100000010000000001)$\\
$T_{16}$&$(1101200100001201000100100100010001001000)$\\
$T_{17}$&$(1101111100000101111010010100000000101000)$\\
$T_{20}$&$(1111111000010111001110000110100000110001)$\\
$T_{24}$&$(1100111111010111101110100010111100000011)$\\
$T_{29}$&$(1210101110111102110111000110020001020111)$\\
$T_{30}$&$(1110111110111111111110100111110111010010)$\\
$T_{33}$&$(1111111101001101111111111111101011111110)$\\
\hline
$T_{14}$&$(1110111000000111001000100100000000100001)$\\
$T_{19}$&$(1102110100100100201010010000010011010100)$\\
$T_{26}$&$(2100101011101100110110011011101011011101)$\\
$T_{35}$&$(1110111111100122000211001211111012012110)$\\
$T_{36}$&$(1111111111111111111111111101111101100111)$\\ 
$T_{38}$&$(1111111012210121110121011102110012101111)$\\
$T_{39}$&$(2111111211110111101111111111211111101110)$\\
\noalign{\hrule height0.8pt}
\end{tabular}
}
\end{center}
\end{table}

% By Lemma~\ref{lem:Gb}, we have the following:
% \begin{equation}\label{eq:3-5}
% d_3(n,4) \le g_3(n,4)-1 \text{ if } 
% n \equiv
% 0,6,10,14,19,23,27,32,36 \pmod{40} 
% \end{equation}
% for $n \ge 4$.
% From Proposition~\ref{prop:4-0} and \eqref{eq:3-5},
% we have the following:

By Lemma~\ref{lem:Gb}~(ii), we have
\begin{equation}\label{eq:3-5}
d_3(n,4) \le g_3(n,4)-1 \text{ if } 
n \equiv
0,6,10,14,19,23,27,32,36 \pmod{40}. 
\end{equation}
Hence, we have the following:

\begin{prop}\label{prop:3-4-1}
\begin{enumerate}
\item 
If $n \equiv 4,5,7,8,11,16,17,20,24,29,30,33 \pmod{40}$ and $n \ge 4$, then
$d_3(n,4)=g_3(n,4)$.
\item 
If  $n \equiv  6,10,14,19,36 \pmod{40}$ and $n \ge 4$, then
$d_3(n,4) = g_3(n,4)-1$. 
\end{enumerate}
\end{prop}

In order to apply Proposition~\ref{prop:main2} to this case,
write $n=40s+t$, where $s \in \ZZ_{\ge 0}$ and
\[
 t \in \{
 0,1,2,3,9,12,13,15,18,21,22,23,25,26,27,28,31,32,34,35,37,38,39\}.
\]
Suppose that
\[
d(s,t)=
\begin{cases}
 g_3(40s+t,4)-1 &\text{ if } t \in\{0,23,27,32\},\\
 g_3(40s+t,4)   &\text{ otherwise.}
\end{cases}
\]
Let $r=r_{3,40s+t,4,d(s,t)}$ be the integer defined in~\eqref{eq:r},
where $r$ is listed in Table~\ref{Tab:3-4-0}.
Note that $d(s,t)$ is written as
$27s+\alpha(t)$, where $\alpha(t)$ is a constant depending on only $t$.
Since $d(s,t)$ satisfies 
the assumption~\eqref{eq:as} in Proposition~\ref{prop:main2},
we have the following:

\begin{prop}\label{prop:3-4-main}
% Let $s'=s'_{3,n,4,d}$ and $r=r_{3,n,4,d}$ be the integers defined in 
% \eqref{eq:s0} and \eqref{eq:r}, respectively.
% For $s \ge s'$,
% if there is no ternary LCD $[3r,4,2r]$ code
% $C_0$ with $d(C_0^\perp) \ge 2$,
% then there is no ternary LCD $[40s+t,4,d]$ code.
If there is no ternary LCD $[3r,4,2r]$ code
with dual distance $d^\perp \ge 2$,
then there is no ternary LCD $[40s+t,4,d(s,t)]$ code
for every integer $s$.
\end{prop}

% \begin{rem}
% For each $n=40s+t$, we list $r$ in Table~\ref{Tab:3-4-0}.
% \end{rem}

%%%%%%%%%%%%%%%%%%%%%%%%%%%%%%
\begin{table}[thb]
\caption{$r$ in Proposition~\ref{prop:3-4-main}}
\label{Tab:3-4-0}
\begin{center}
{\small
%{\footnotesize
%{\scriptsize
%{\tiny
\begin{tabular}{c|c||c|c||c|c||c|c}
\noalign{\hrule height0.8pt}
$n$ & $r$ &$n$ & $r$ &$n$ & $r$ &$n$ & $r$ \\
\hline
$40s   $&40&$40s+13$&31&$40s+25$&35&$40s+34$&38\\
$40s+ 1$&27&$40s+15$&45&$40s+26$&22&$40s+35$&25\\
$40s+ 2$&54&$40s+18$&46&$40s+27$&49&$40s+37$&39\\
$40s+ 3$&81&$40s+21$&47&$40s+28$&36&$40s+38$&26\\
$40s+ 9$&43&$40s+22$&34&$40s+31$&37&$40s+39$&13\\
$40s+12$&44&$40s+23$&61&$40s+32$&64&&\\
\noalign{\hrule height0.8pt}
\end{tabular}
}
\end{center}
\end{table}

By Proposition~\ref{prop:3-4-main}, we examine the nonexistence
of a ternary LCD $[3r,4,2r]$ code for $r$ in Table~\ref{Tab:3-4-0}.
As described in Remark~\ref{rem},
it is possible to find
representatives of all equivalence classes of ternary LCD
$[n,4,d]$ codes with dual distances $d^\perp \ge 2$
as the codes $C_{3,4}(m)$,
by considering all vectors $m=(m_1,m_2,\ldots,m_{40}) \in
\mathbb{Z}_{\ge 0}^{40}$ 
such that $n=\sum_{i=1}^{40}m_i$ and~\eqref{eq:mi}
for a given set of parameters $n,d$.
Moreover, 
we may assume without loss of generality that 
\[
m_i \ge 1\ (i \in \{1,2,5,14\})
\text{ and } \sum_{i \in {\mathcal S}} m_i = d,
\]
where 
%$\Gamma=\{j+k \mid j \in \{1,6,10,15,20,24,28,33,37\}, k\in \{0,2,3\}\}$.
$\mathcal S$ is the support of the first row of the matrix $S_{3,4}$.
In this way,
our exhaustive computer search shows that
there is no ternary LCD $[3r,4,2r]$ code with dual distance $d^\perp \ge 2$
for only
$r \in \{13,22,25,26\}$,
%% 13^c, \mathbf{22,25,26^c,27,31,34,35,36,37,38,}\\
%% \mathbf{39,40,43,44,45,46,47,54,61,64,81}
due to the computational complexity.
% For reference, 
% the time required for the computer search of $r=22$, 
% which corresponds to a single core of
% a computer with Intel i7,
% is approximately {\bf ???? days}.
For reference,
the time required for the computer search of $r=22$, 
which corresponds to a single core of
a computer with Intel i7,
is approximately 709 days.
%% This was done by considering all possible vectors $m$ in
%% generator matrices $G_{3,4}(m)$ of form~\eqref{eq:Gqkm}
%% under the condition~\eqref{eq:mi}.
For the remaining cases $r$, 
an exhaustive search remains a computational challenge.

Therefore, by Proposition~\ref{prop:3-4-main},
we have
%\begin{align*}
%d_3(40s+38,4) &\le 27s+24,\\
%d_3(40s+39,4) &\le 27s+25.
%\end{align*}
\[
d_3(40s+t,4)\le
g_3(40s+t,4)-1 \text{ if } t \in \{26,35,38,39\}.
\]
From Proposition~\ref{prop:4-0}~(ii), we have the following:

\begin{prop}For $n \ge 4$,
\[
 d_3(n,4)=g_3(n,4)-1  \text{ if } n
 \equiv 26,35,38,39 \pmod{40}.
\]
\end{prop}

We summarize in the following theorem the largest minimum
weights $d_3(n,4)$.

\begin{thm}\label{thm:3-4} 
If
 $n \equiv 4,5,7,8,10,11,14,16,17,19,20,24,26,29,30,33,35,36$,
 $38,39 \pmod{40}$
and $n \ge 4$, then
\[
d_3(n,4)=\left\lfloor \frac{27n}{40}\right\rfloor-1.
\]
If $n \equiv 6  \pmod{40}$
and $n \ge 4$, then
\[
d_3(n,4)=\left\lfloor \frac{27n}{40}\right\rfloor-2.
\]
\end{thm}

Finally, we consider the remaining lengths.
It is known that $d_3(9,4)=4$~\cite[Table~4]{AH-C}.
As described in Section~\ref{sec:1}, 
we determine $d_3(n,k)$ for $n \in \{11,12,\ldots,19\}$
in Section~\ref{sec:d}.
From Table~\ref{Tab:C}, we have
\[
d_3(12,4)=6, d_3(13,4)=7, d_3(15,4)=8,
d_3(18,4)=10.
\]
By considering vectors $m$
in generator matrices $G_{3,4}(m)$ of form~\eqref{eq:Gqkm}
such that $n=\sum_{i=1}^{40}m_i$ and~\eqref{eq:mi},
we found a ternary LCD $[n,4,d]$ code $T_n$ for
\[
(n,d) \in 
\left\{
\begin{array}{l}
 (21,12),(22,13),(23,13),(25,15),(27,16),(28,17), (31,19),\\
 (32,19),(34,21),(37,23),(40,25),(41,26), (42,26),(43,26)
\end{array}
\right\}
\]
as codes $C_{3,4}(m)$.
The vectors $m$ are listed in Table~\ref{Tab:3-4-m3}.
%By Lemma~\ref{lem:Gb}, we have the following:
%\begin{equation*}
%d_3(n,4) \le g_3(n,4)-1 \text{ if } 
%n \equiv
%0,23,27,32 \pmod{40} 
%\end{equation*}
%for $n \ge 4$.
Hence, by
Lemma~\ref{lem:extend} and~\eqref{eq:3-5}, 
we have the following:

\begin{prop}
Suppose that $n \ge 21$. 
If $n \equiv 1, 13, 22, 25, 28, 31, 34, 37 \pmod{40}$, then
\[
 d_3(n,4)=
 \left\lfloor \frac{27n}{40}\right\rfloor \text{ or }
 \left\lfloor \frac{27n}{40}\right\rfloor-1.
\]
If $n \equiv 0, 2,3,9, 12, 15, 18, 21, 23, 27, 32 \pmod{40}$, then
\[
 d_3(n,4)=
 \left\lfloor \frac{27n}{40}\right\rfloor-1 \text{ or }
 \left\lfloor \frac{27n}{40}\right\rfloor-2.
\]
 \end{prop}

%%%%%%%%%%%%%%%%%%%%%%%%%%%%%%
\begin{table}[thbp]
\caption{Codes $T_n$ $(n \in
 \{21,22,23,25,27,28,31,32,34,37, 40,41,42,43\})$}
\label{Tab:3-4-m3}
\begin{center}
{\small
%{\footnotesize
%{\scriptsize
%{\tiny
\begin{tabular}{c|c}
\noalign{\hrule height0.8pt}
Codes & Vectors $m$ \\
\hline
$T_{21}$&$(1110110101000111100111011000200001001010)$\\
$T_{22}$&$(1110100101110100000121010101100011101100)$\\
$T_{23}$&$(1120111000200111000110010101010011001110)$\\
$T_{25}$&$(2100111001111101000111011101010010200200)$\\
%$T_{26}$&$(2100101011101100110110011011101011011101)$\\
$T_{27}$&$(2210300120000101110110010110010011010111)$\\
$T_{28}$&$(1110121101110101100011020101100011101210)$\\
$T_{31}$&$(1110111102200112100001011111100021111110)$\\
$T_{32}$&$(2210112002010112100210011100120021100110)$\\
$T_{34}$&$(2210111101100111000122021101211011101020)$\\
$T_{37}$&$(1100111112210112110221111102110001102110)$\\
$T_{40}$&$(2220121001110111100112012112201011202210)$\\
$T_{41}$&$(2210110111201121011121011110111022112111)$\\
$T_{42}$&$(2220110202310122200111012101200031101120)$\\
$T_{43}$&$(3110122202210101200012023003210022001120)$\\
 \noalign{\hrule height0.8pt}
\end{tabular}
}
\end{center}
\end{table}
%%%%%%%%%%%%%%%%%%%%%%%%%%%%%%

%%%%%%%%%%%%%%%%%%%%
\section{Ternary LCD codes of lengths up to 20}\label{sec:d}

For $n\in \{1,2,\ldots,10\}$,
the classification of ternary LCD $[n,k]$ codes
was done in~\cite{AH-C}.
In this section, we determine
the largest minimum weights $d_3(n,k)$ among all ternary LCD $[n,k]$
codes for $n \in \{11,12,\ldots,19\}$.
We also determine the largest minimum weights $d_3(20,k)$
with $4$ exceptions.

We employ two methods for constructing ternary codes.
\begin{itemize}
\item Method I:
Every ternary $[n,k,d]$ code is equivalent to a code
with generator matrix of form:
\[
\left(
\begin{array}{cc}
I_k & A
\end{array}
\right),
\]
where $A$ is a $k \times (n-k)$ matrix and 
$I_k$ denotes the identity matrix of order $k$.
Let $r_i$ be the $i$-th row of $A$.
Here, we may assume that $A$ satisfies the following conditions:
\begin{itemize}
\item[\rm (a)]
$r_1=(\0_{n-k-d+1}, \1_{d-1})$,
\item[\rm (b)]
% $\wt(r_i) \ge d-1$,
the weight of $r_i$ is at least $d-1$ $(i \in\{2,3,\ldots,k\})$,
\item[\rm (c)]
the first nonzero element of $r_i$ is $1$  $(i \in\{2,3,\ldots,k\})$,
\item[\rm (d)]
$r_1 < r_2 < \cdots < r_k$ if $d \ge 3$ and
$r_1 \le r_2 \le \cdots \le r_k$ if $d \le 2$,
\end{itemize}
where 
% $\wt(x)$ denotes the weight of $x$, and
we consider some order $<$ on
the set of vectors of length $n-k$.
% consider a natural order on $\FF_3=\{0,1,2\}$
% as follows $0 <1 <2$.
The set of matrices $A$ is constructed, row by row, under
the assumption that the minimum weight of the ternary 
$[n+m-k,m]$ code with generator matrix of form:
\[
\left(
\begin{array}{ccccc}
      & r_1 \\
I_{m} & \vdots   \\
      & r_m \\
\end{array}
\right)
\]
is at least $d$ for each $m \in \{2,3,\ldots,k-1\}$.
%It is obvious that all inequivalent codes, which must be checked, 
%can be obtained in this way.
It is obvious that the set of the ternary $[n,k,d]$
codes obtained by this method
contains a set of all inequivalent ternary $[n,k,d]$ codes.

\item Method II:
Let $C$ be a ternary code.
Let $C(t)$ be the set of all codewords which are $0$ in a fixed
coordinate $t$.
The ternary code obtained from $C(t)$ by deleting the
coordinate $t$  in each codeword is called a {\em shortened code} of $C$.
% A {\em shortened code} of a ternary code $C$ is the set of all codewords
% in $C$ which are $0$ in a fixed coordinate with that
% coordinate deleted.
A shortened code of a ternary $[n,k,d]$ code with $d \ge 2$
is a ternary $[n-1,k,d]$ code if the deleted coordinate
is zero in all codewords and a ternary $[n-1,k-1,d']$
code with $d' \ge d$ otherwise.
%% A ternary $[n,k,d]$ code gives $n$ shortened codes,
%% and at least $k$ codes among them are ternary $[n-1,k-1,d']$ codes
%% with $d' \ge d$.
%% Hence,
By considering the inverse operation of shortening,
every ternary $[n,k,d]$ code with $d \ge 2$ is constructed from some
ternary $[n-1,k-1,d']$ code with $d' \ge d$.
It is obvious that the set of the ternary $[n,k,d]$
codes obtained by this method
contains a set of all inequivalent ternary $[n,k,d]$ codes.
This method is useful for small $k$.
\end{itemize}

We describe how our computer calculation determined
the minimum weights $d_3(n,k)$.
Let $d_3^{\text{all}}(n,k)$ denote the largest minimum weight
among all ternary $[n,k]$ codes.
% One can find the current information on $d_3^{\text{all}}(n,k)$ 
% in~\cite{Br}.
For a given pair $(n,k)$,
we checked whether there is a ternary LCD $[n,k,d_3^{\text{all}}(n,k)]$
code or not, by using one of the above methods.
If there is no ternary LCD $[n,k,d_3^{\text{all}}(n,k)]$ code,
then
we checked whether there is a ternary LCD $[n,k,d_3^{\text{all}}(n,k)-1]$
code or not.
By continuing this process, we determined the minimum weights $d_3(n,k)$
for $n \in \{11,12,\ldots,19\}$.
%%%%%%%%%%
% In this process, Lemma~\ref{lem:Gb} is also used for minimum weight $6$.
% It is known~\cite{Br} that
% $d_3^{\text{all}}(k+6,k)=6$
% and a ternary $[k+6,k,6]$ code meets the
% Griesmer bound for $k \in \{5,6\}$.
% By Lemma~\ref{lem:Gb}, there is no ternary LCD $[k+6,k,6]$ code
% for $k \in \{5,6\}$.
%%%%%%%%%%
We also determined the minimum weights $d_3(20,k)$
with $4$ exceptions.
The largest minimum weights $d_3(n,k)$ are known
for $k \in \{1,2,3,n-1,n\}$.
Also, $d_3(n,k)$ are determined in the next section for
$k \in \{n-4,n-3,n-2\}$.
In Table~\ref{Tab:C}, we only list $d_3(n,k)$ for $k \in \{4,5,\ldots,n-5\}$.
For the parameters in the table, 
a ternary LCD code can be obtained electronically from
\url{http://www.math.is.tohoku.ac.jp/~mharada/Paper/LCD3.txt}.
For the parameters marked by $*$ in the table,
ternary LCD codes with the parameters
can be found in~\cite[Table~3]{Sok}.

%%%%%%%%%%%%%%%%%%%%%%%%%%%%%%
\begin{table}[thb]
\caption{$d_3(n,k)$ $(n \in \{11,12,\ldots,20\})$}
\label{Tab:C}
\begin{center}
{\small
%{\footnotesize
%{\scriptsize
%{\tiny
%\begin{tabular}{c|cccccccccccccccccccccc}
\begin{tabular}{c|lllllllllllllllllllllll}
\noalign{\hrule height0.8pt}
$n\backslash k$  
%    &  4&  5&  6&  7&  8&  9& 10& 11& 12& 13& 14& 15 \\
&\multicolumn{1}{c}{4}&\multicolumn{1}{c}{5}&\multicolumn{1}{c}{6}
&\multicolumn{1}{c}{7}&\multicolumn{1}{c}{8}&\multicolumn{1}{c}{9}
&\multicolumn{1}{c}{10}&\multicolumn{1}{c}{11}&\multicolumn{1}{c}{12}
&\multicolumn{1}{c}{13}&\multicolumn{1}{c}{14}&\multicolumn{1}{c}{15}\\
 \hline
$11$&$ 6  $&$ 5$&$ 4$&   &  &  &  &   &   &   &   &     \\
$12$&$ 6^*$&$ 5$&$ 5^*$& $ 4$&  &  &  &  &   &   &   &     \\
$13$&$ 7  $&$ 6$&$ 6$  & $ 5$& $ 4$&  &  &  &  &   &   &     \\
$14$&$ 8^*$&$ 7$&$ 6^*$& $ 6$& $ 5^*$& $ 4$&  &  &  &  &   &     \\
$15$&$ 8  $&$ 8$&$ 7$  & $ 6$& $ 5$& $ 4$& $ 4$&  &  &  &  &     \\
$16$&$ 9^*$&$ 8$&$ 7^*$& $ 6$& $ 6^*$& $ 5$& $ 4^*$& $ 4$&  &  &  &    \\
$17$&$10  $&$ 9$&$ 8$& $ 7$& $ 6$& $ 6$& $ 5$& $ 4$& $ 4$&  &  &    \\
$18$&$10^*$&$ 9$&$ 9$& $ 8$& $ 7^*$& $ 6$& $ 6$& $ 5$& $ 4^*$& $ 4$&  &    \\
$19$&$11$&$10$&$ 9$&$8$&$ 8$& $ 7$& $ 6$& $ 6$& $ 5$& $ 4$& $ 4$&    \\
$20$&$12$&$11$&$10$&$8, 9$&$8^*$&$7, 8$&$7$&$6$&$5^*, 6$&$5$&$4^*$&$3, 4$\\
%11&  6&  5&  4&  &  &  &  &   &   &   &   &     \\
%12&  6&  5&  5&  4&  &  &  &  &   &   &   &     \\
%13&  7&  6&  6&  5&  4&  &  &  &  &   &   &     \\
%14&  8&  7&  6&  6&  5&  4&  &  &  &  &   &     \\
%15&  8&  8&  7&  6&  5&  4&  4&  &  &  &  &     \\
%16&  9&  8&  7&  6&  6&  5&  4&  4&  &  &  &    \\
%17& 10&  9&  8&  7&  6&  6&  5&  4&  4&  &  &    \\
%18& 10&  9&  9&  8&  7&  6&  6&  5&  4&  4&  &    \\
%19& 11&  10&  9&  8, 9&  8&  7&  6&  6&  5&  4&  4&    \\
%20& 12&11& 10& 8, 9& 8, 9 & 7, 8 & 7& 6& 5, 6 & 5 & 4& 3, 4\\
%%%%%%%%%%%%%%%%%
 %11&  7&  6&  6&  5&  4&  &  &  &  &   &   &   &   &     \\
%12&  8&  7&  6&  5&  5&  4&  &  &  &  &   &   &   &     \\
%13&  9&  8&  7&  6&  6&  5&  4&  &  &  &  &   &   &     \\
%14& 10&  8&  8&  7&  6&  6&  5&  4&  &  &  &  &   &     \\
%15& 10&  9&  8&  8&  7&  6&  5&  4&  4&  &  &  &  &     \\
%16& 11& 10&  9&  8&  7&  6&  6&  5&  4&  4&  &  &  &    \\
%17& 12& 11& 10&  9&  8&  7&  6&  6&  5&  4&  4&  &  &    \\
%18& 13& 11& 10&  9&  9&  8&  7&  6&  6&  5&  4&  4&  &    \\
%19& 13& 12& 11&  10&  9&  8, 9&  8&  7&  6&  6&  5&  4&  4&    \\
%20& 14& 12& 12&11& 10& 8, 9& 8, 9 & 7, 8 & 7& 6& 5, 6 & 5 & 4& 3, 4\\
\noalign{\hrule height0.8pt}
\end{tabular}
}
\end{center}
\end{table}
%%%%%%%%%%%%%%%%%%%%%%%%%%%%%%%

Now we emphasize that there is a ternary LCD $[n,k,d]$ code  $C_{n,k,d}$
for
\[
(n,k,d) \in
\left\{
\begin{array}{l}
  (11,5,5),
  (12,6,5),
  (13,6,6),
  (13,7,5),
  (14,7,6), \\
  (14,8,5), 
  (15,5,8),
  (17,5,9),
  (18,6,9) 
\end{array}
\right\}.
\]
The codes $C_{n,k,d}$ have generator matrices
$
\left(\begin{array}{cc}
I_k & M_{n,k,d} 
\end{array}\right)
$ and the matrices $M_{n,k,d}$ are listed in Figure~\ref{Fig:E}.

\begin{thm}
\begin{enumerate}
\item 
For a nonnegative integer $s$, 
\[
\begin{array}{ll}
d_3(121s+17,5)=81s+9, &d_3(364s+13,6)=243s+6,  \\ 
d_3(364s+18,6)=243s+9,&d_3(1093s+14,7)=729s+6.
\end{array}
\]
\item 
For a nonnegative integer $s$, 
\[
\begin{array}{ll}
d_3(  40s+ 10, 4)=   27s+ 5,&
d_3( 121s+ 11, 5)=   81s+ 5,\\
d_3( 121s+ 15, 5)=   81s+ 8,&
d_3( 364s+ 12, 6)=  243s+ 5,\\
d_3(1093s+ 13, 7)=  729s+ 5,&
d_3(3280s+ 14, 8)= 2187s+ 5.
\end{array}
\]
\end{enumerate}
\end{thm}
\begin{proof}
\begin{enumerate}
\item
From Table~\ref{Tab:C}, 
there is a ternary LCD $[n,k,d]$ code
for
\[
(n,k,d) \in \{(13,6,6), (14,7,6), (17,5,9), (18,6,9)\}.
\]
By Lemma \ref{lem:extend},
there is a ternary LCD 
$[n+[k]_3\cdot s,k,d+3^{k-1}s]$ code for every positive integer $s$.
The assertion follows from the Griesmer bound.
     
\item 
Let $s$ be a positive integer.
By the Griesmer bound, we have
\[
\begin{array}{ll}
d_3(  40s+ 10, 4)\le   27s+ 6,&
d_3( 121s+ 11, 5)\le   81s+ 6,\\
d_3( 121s+ 15, 5)\le   81s+ 9,&
d_3( 364s+ 12, 6)\le  243s+ 6,\\
d_3(1093s+ 13, 7)\le  729s+ 6,&
d_3(3280s+ 14, 8)\le 2187s+ 6.
\end{array}
\]
For 
\begin{equation*}
(n,k,d) \in
\left\{
\begin{array}{ll}
(  40s+ 10, 4,   27s+ 6),&
( 121s+ 11, 5,   81s+ 6),\\
( 121s+ 15, 5,   81s+ 9),&
( 364s+ 12, 6,  243s+ 6),\\
(1093s+ 13, 7,  729s+ 6),&
(3280s+ 14, 8, 2187s+ 6)
\end{array}
\right\},
\end{equation*}
each ternary $[n,k,d]$ code meets the Griesmer bound.
Since $d$ is a multiple of $3$,
by Lemma~\ref{lem:Gb}~(ii), it is not LCD.
Hence, we have
\[
\begin{array}{ll}
d_3(  40s+ 10, 4)\le   27s+ 5,&
d_3( 121s+ 11, 5)\le   81s+ 5,\\
d_3( 121s+ 15, 5)\le   81s+ 8,&
d_3( 364s+ 12, 6)\le  243s+ 5,\\
d_3(1093s+ 13, 7)\le  729s+ 5,&
d_3(3280s+ 14, 8)\le 2187s+ 5.
\end{array}
\]
From~\cite[Table~4]{AH-C} and Table~\ref{Tab:C}, it is known that
\begin{multline*}
d_3(10, 4)= 5,
d_3(11, 5)= 5,
d_3(15, 5)= 8,\\
d_3(12, 6)= 5,
d_3(13, 7)= 5,
d_3(14, 8)= 5.
\end{multline*}
By Lemma~\ref{lem:extend},
there is a ternary LCD
$[n,k,d]$ code for
\[
(n,k,d) \in
\left\{
\begin{array}{ll}
(  40s+ 10, 4,   27s+ 5),&
( 121s+ 11, 5,   81s+ 5),\\
( 121s+ 15, 5,   81s+ 8),&
( 364s+ 12, 6,  243s+ 5),\\
(1093s+ 13, 7,  729s+ 5),&
(3280s+ 14, 8, 2187s+ 5)
\end{array}
\right\}.
\]
\end{enumerate}
This completes the proof.
\end{proof}

\begin{rem}
% Note that $d_3(n,k)$ is known for $n \le 19$.
Only the parameters
$[ 11, 5, 5 ]$,
$[ 12, 6, 5 ]$,
$[ 13, 7, 5 ]$,
$[ 14, 8, 5 ]$,
$[ 14, 4, 8 ]$ and
$[ 15, 5, 8 ]$
are
parameters $[n,k,d_3(n,k)+1]$
meeting the Griesmer bound and
satisfying the
assumption of Lemma~\ref{lem:Gb}~(ii) for $n \le 19$ and $5 \le k \le n-5$.
\end{rem}

\begin{figure}[htbp]
\centering
{\small
%%{\footnotesize
\begin{align*}
\begin{array}{l}
M_{11,5,5}=\left(\begin{array}{c}
001111\\
012110\\
111100\\
121010\\
120101
\end{array}\right),
M_{12,6,5}=\left(\begin{array}{c}
001111\\
012110\\
111100\\
121010\\
120101\\
122112
\end{array}\right),
M_{13,6,6}=\left(\begin{array}{cccccccccccccccc}
0011111\\
0122110\\
1111100\\
1220101\\
1110011\\
1021012
\end{array}\right),
\\
M_{13,7,5}=\left(\begin{array}{c}
001111\\
012110\\
111100\\
121010\\
112001\\
120021\\
122112
\end{array}\right),
M_{14,7,6}=\left(\begin{array}{cccccccccccccccc}
0011111\\
0122110\\
1111100\\
1220101\\
1110011\\
1021012\\
1212112
\end{array}\right),
M_{14,8,5}=\left(\begin{array}{c}
001111\\
011122\\
012110\\
111100\\
121010\\
112001\\
120021\\
122112
\end{array}\right),
\\
M_{15,5,8}=
\left(\begin{array}{cccccccccccccccc}
 0 0 0 1 1 1 1 1 1 1\\
 0 1 1 2 2 1 1 1 0 0\\
 1 2 1 1 1 1 0 0 1 0\\
 1 2 2 2 0 1 2 0 0 1\\
 1 1 0 0 1 0 2 1 2 2
\end{array}\right),
M_{17,5,9}=
\left(\begin{array}{cccccccccccccccc}
000011111111\\
001122211100\\
112121110000\\
121111001010\\
122202011201
\end{array}\right),
\\
M_{18,6,9}=\left(\begin{array}{cccccccccccccccc}
000011111111\\
111101111122\\
122112110000\\
121210201100\\
120121202010\\
112011221212
\end{array}\right)
\end{array}
\end{align*}
 \caption{Matrices $M_{n,k,d}$}
\label{Fig:E}
}
\end{figure}

%%%%%%%%%%%%%%%%%%%%%%%%%%%%%%%%%%
\section{Large dimensions}\label{sec:large}

In this section,
we examine the largest minimum weights $d_2(n,n-i)$ and 
$d_3(n,n-i)$ for small $i$.
In particular, we completely determine
$d_2(n,n-5)$ and $d_3(n,n-i)$ $(i\in\{2,3,4\})$ for arbitrary $n$.

%%%%%%%%%%%%%%%%%%%%%%%%%%
\subsection{Binary LCD codes of dimension $n-5$}

The minimum weights $d_2(n,k)$ were determined for
$k=n-1$ in~\cite{DKOSS} and $k \in \{n-2,n-3,n-4\}$ in~\cite{AH}.

The following lemma is a key idea for the determination of 
$d_2(n,n-i)$ for small $i$.

\begin{lem}[{\cite[Theorem~3]{GKLRW}}]
\label{lem:F2-n-i}
Let $i$ be an integer with $2 \le i < n$.
If $n \ge 2^i$, then
$d_2(n,n-i) = 2$.
 \end{lem}

\begin{prop}\label{prop:F2-n-5}
For $n \ge 6$, 
\[
d_2(n,n-5)=
\begin{cases}
5& \text{ if } n=6, \\
4& \text{ if } n \in \{7,9,11\},\\
3& \text{ if } n \in \{8,10,12,13,\ldots,26\},\\
2& \text{ if } n \in \{27,28,\ldots\}.\\
\end{cases}
\]
\end{prop}
\begin{proof}
If $n \ge 32$, then $d_2(n,n-5)=2$ by Lemma~\ref{lem:F2-n-i}.
It is known that
$d_2(n,n-5)=3$ for $n\in\{8,10,12,13,\ldots,24\}$,
$d_2(n,n-5)=4$ for $n\in \{7,9,11\}$ and
 $d_2(6,1)=5$~\cite[Table~15]{AH}, \cite[Table~1]{GKLRW}
 and~\cite[Table~3]{HS}.
Let $C_{26}$ be the binary $[26,21]$ code with parity-check matrix
$H_{26}=
\left(
\begin{array}{ccccc}
I_5&  M_{26}  \\
\end{array}
\right)
$, where
\[
M_{26}=
\begin{pmatrix}
0 0 0 0 0 0 0 1 1 1 1 1 1 1 1 1 1 1 1 1 1\\
0 1 1 1 1 1 1 0 0 0 0 0 0 1 1 1 1 1 1 1 1\\
1 0 0 1 1 1 1 0 0 1 1 1 1 0 0 0 0 1 1 1 1\\
1 1 1 0 0 1 1 1 1 0 0 1 1 0 0 1 1 0 0 1 1\\
0 0 1 0 1 0 1 0 1 0 1 0 1 0 1 0 1 0 1 0 1
 \end{pmatrix}.
\]
Let $H_{25}$ be the matrix obtained from $H_{26}$ by deleting
the last column.
Let $C_{25}$ be the binary $[25,20]$ code with parity-check matrix $H_{25}$.
We verified that $C_{25}$ and $C_{26}$ are binary LCD codes with
parameters $[25,20,3]$ and $[26,21,3]$, respectively.
Our exhaustive computer search shows that there is no binary
$[n,n-5,3]$ code for $n \in \{27,28,\ldots,31\}$.
This was done by the method which is obtained by applying
Method~I in Section~\ref{sec:d}.
\end{proof}

%%%%%%%%%%%%%%%%%%%%%%%%%%
\subsection{Ternary LCD codes of dimensions $n-2,n-3,n-4$}

The classification of ternary LCD $[n,n-1]$ codes was done in~\cite{AH-C}.
In this subsection, we determine the largest minimum weights
$d_3(n,n-i)$ for arbitrary $n$ and $i \in \{2,3,4\}$.

The following lemma is a key idea for the determination of 
$d_3(n,n-i)$ for small $i$.
% The lemma is an analogue of Lemma~\ref{lem:F2-n-i}
% for binary LCD codes.

\begin{lem}\label{lem:n-i}
Suppose that $i$ is an integer with $2 \le i \le n-1$.
\begin{enumerate}
\item 
There is a ternary LCD $[n,n-i,2]$ code.
\item 
If $n > \frac{3^i-1}{2}$, then $d_3(n,n-i) = 2$.
\end{enumerate}
\end{lem}
\begin{proof}
Let $C$ be a ternary $[n,n-i]$ code with parity-check matrix $H$, where
% \[
% H=
% \begin{cases}
% \left(
% \begin{array}{cccccccccc}
% 1 &0       & 1&1&1& \cdots & 1 \\
% 0 &1       & 1&1&0& \cdots & 0 \\
% \end{array}
% \right) & \text{ if } i = 2, \\
% \left(
% \begin{array}{cccccccccc}
%  &       & 1&1&1& \cdots & 1 \\
%  &       & 1&1&0& \cdots & 0 \\
%  &I_i    & 0&0&0& \cdots & 0 \\
%  &       & \vdots& \vdots & \vdots & & \vdots\\
%  &       & 0&0&0& \cdots & 0 \\
% \end{array}
% \right) & \text{ if } i \ge 3.
% \end{cases}
% \]
\[
H=
\begin{cases}
\left(
\begin{array}{cccccccccc}
1 &0       & 1&1&1& \cdots & 1 \\
0 &1       & 1&1&0& \cdots & 0 \\
\end{array}
\right) & \text{ if } i = 2, \\
\left(
\begin{array}{cccccccccc}
 &       & 1&1&1& \cdots & 1 \\
 & I_i   & 1&1&0& \cdots & 0 \\
 &       & \0_{i-2}^T&\0_{i-2}^T&\0_{i-2}^T& \cdots & \0_{i-2}^T \\
\end{array}
\right) & \text{ if } i \in \{3,4,\ldots,n-1\}.
\end{cases}
\]
% It is known that a ternary $[n,k]$ code $C$ is LCD if and only if
% $GG^T$ is nonsingule for any generator matrix $G$ of $C$,
Since 
\[
HH^T=
\begin{cases}
\left(
\begin{array}{cccccccccc}
n-1 & 2 \\
2     & 0 
\end{array}
\right) & \text{ if } i=2, \\
\left(
\begin{array}{cccccccccc}
n-i+1 & 2 &  & \0_{i-2} &  \\
2     & 0 &  & \0_{i-2} &  \\
\0_{i-2}^T&\0_{i-2}^T& & I_{i-2} & \\
\end{array}
\right) & \text{ if } i \in \{3,4,\ldots,n-1\},
\end{cases}
\]
% \[
% HH^T=
% \left(
% \begin{array}{cccccccccc}
% n-i+1 & 2 & 0 & \cdots & 0 \\
% 2     & 0 & 0 & \cdots & 0 \\
% 0     & 0 &   &        & \\
% \vdots&\vdots& & I_{i-2} & \\
% 0     & 0 &   &        & \\
% \end{array}
% \right) \text{ if $i \ge 3$,}
% \]
the code $C$ is LCD.
By the construction, it is trivial that $C$ has minimum weight $2$.
This proves the assertion (i).
 
Suppose that there is a ternary $[n,n-i,d]$ code.
By the sphere-packing bound, if $d \ge 3$, then
$n \le \frac{3^i-1}{2}$.
This proves the assertion (ii).
\end{proof}

%%%%%%%%%%%%%%%%%%%%%%%%%%%%%%%%%%%%

\begin{prop}\label{prop:n-2}
For $n \ge 3$,
\[
d_3(n,n-2)=2.
\]
\end{prop}
\begin{proof}
By~\cite[Proposition~5]{AH-C}, $d_3(3,1)=2$.
It is known that $d_3(4,2)=2$~\cite[Table~4]{AH-C}.
If $n \ge 5$, then $d_3(n,n-2)=2$ by Lemma~\ref{lem:n-i}~(ii).
\end{proof}

\begin{prop}\label{prop:n-3}
For $n \ge 4$,
\[
d_3(n,n-3)=
\begin{cases}
4& \text{ if } n=4, \\
3& \text{ if } n \in \{5,6,7,8,9,10\}, \\
2& \text{ if } n \in \{11,12,\ldots\}.\\
\end{cases}
\]
\end{prop}
\begin{proof}
By~\cite[Proposition~5]{AH-C}, $d_3(4,1)=4$.
It is known that
$d_3(n,n-3)=3$ for $n \in \{5,6,7,8,9,10\}$~\cite[Table~4]{AH-C}.
If $n \ge 14$, then $d_3(n,n-3)=2$ by Lemma~\ref{lem:n-i}~(ii).
It is known that 
$d^{\text{all}}_3(n,n-3)=3$ if $n \in \{11,12,13\}$ (see~\cite{Br}).
Our exhaustive computer search shows
that no ternary $[n,n-3,3]$ code is LCD for $n \in \{11,12,13\}$,
by using Method~I in Section~\ref{sec:d}.
By Lemma~\ref{lem:n-i}~(i),
there is a ternary LCD $[n,n-3,2]$ code for $n
 \in \{11,12,13\}$.
The result follows.
\end{proof}

\begin{prop}\label{prop:n-4}
For $n \ge 5$,
\[
d_3(n,n-4)=
\begin{cases}
5& \text{ if } n=5, \\
4& \text{ if } n \in \{6,7,8\}, \\
3& \text{ if } n \in \{9,10,\ldots,36\}, \\
2& \text{ if } n \in \{37,38,\ldots\}.\\
\end{cases}
\]
\end{prop}
\begin{proof}
By~\cite[Proposition~5]{AH-C}, $d_3(5,1)=5$.
It is known that
$d_3(n,n-4)=4$ for $n \in \{6,7,8\}$ and
$d_3(n,n-4)=3$ for $n \in \{9,10\}$~\cite[Table~4]{AH-C}.
If $n \ge 40$, then $d_3(n,n-4)=2$ by Lemma~\ref{lem:n-i}~(ii).

It is known that 
$d^{\text{all}}_3(n,n-4) \le 3$ if $n\in \{11,12,\ldots,39\}$ (see~\cite{Br}).
Let $C_{36}$ be the ternary $[36,32]$ code with generator matrix
$
\left(
\begin{array}{ccccc}
I_{32} &  M_{36}  \\
\end{array}
\right),
$
where 
\[
M_{36}^T=
\left(
\begin{array}{ccccccccccccccc}
11100101110101111111011011011111\\
00210111220110102021120111121002\\
21222000011110210011122021200122\\
00001110001001011211111212122222
\end{array}
\right).
\]
We define the matrices $M_{i+4}$ ($i=31,30,\ldots,7$) by
deleting the last $32-i$ rows of $M_{36}$.
Then let $C_{i+4}$ ($i=7,8,\ldots,31$) be the ternary $[i+4,i]$ code 
with generator matrix
$
\left(
\begin{array}{ccccc}
I_{i} &  M_{i+4}  \\
\end{array}
\right).
$
We verified that $C_n$ is a ternary LCD $[n,n-4,3]$ code
for $n \in \{11,12,\ldots,36\}$.
In addition, our exhaustive computer search shows
that no ternary $[n,n-4,3]$ code is LCD for $n \in \{37,38,39\}$,
by using Method~I in Section~\ref{sec:d}.
By Lemma~\ref{lem:n-i}~(i),
there is a ternary LCD $[n,n-4,2]$ code for $n
\in \{37,38,39\}$.
The result follows.
\end{proof}

%%%%%%%%%%%%%%%%%%%%
\bigskip
\noindent
{\bf Acknowledgment.}
This work was supported by JSPS KAKENHI Grant Number 19H01802.

%%%%%%%%%%%%%%%%%%%  References  %%%%%%%%%%%%%%%%%%%%%%%%

\end{document}